\documentclass{ifacconf}

\makeatletter
\let\old@ssect\@ssect 
\makeatother

\let\theoremstyle\relax
\usepackage{graphicx}      
 \graphicspath{ {./CPHS2020/} }
\usepackage{hyperref}       
\usepackage{url}            
\usepackage{booktabs}       
\usepackage{paralist}
\usepackage{amsfonts}       
\usepackage{nicefrac}       
\usepackage{microtype}      
\usepackage{lipsum}
\usepackage{enumitem}
\usepackage{setspace}

\usepackage{booktabs}
\usepackage{array}

\usepackage{tgheros} 
\usepackage{tgtermes} 
\usepackage{anyfontsize}

\usepackage{enumitem}
 \usepackage{subfigure}
\usepackage{siunitx}

\usepackage{caption}
\DeclareCaptionFont{cfs}{\fontsize{8.5}{10.25}\selectfont}
\DeclareCaptionLabelSeparator{vline}{\;|\;}

 \usepackage{amsmath}
\usepackage{amssymb,upref}
 \usepackage{amsthm}

\newcommand{\di}{\mathrm{i}} 
\newcommand{\argminF}{\mathop{\mathrm{argmax}}\limits}   

\DeclareMathOperator{\Tr}{Tr}
\newcommand{\figurecaption}[2]{\caption[#1]{\textbf{#1.} #2}}

\usepackage{natbib}        

\usepackage{breqn} 

\usepackage{hyperref}
\usepackage{xcolor} 

\newtheorem{theorem}{Theorem}[section]

\newtheorem{proposition}[theorem]{Proposition}

\makeatletter
\def\@ssect#1#2#3#4#5#6{%
  \NR@gettitle{#6}
  \old@ssect{#1}{#2}{#3}{#4}{#5}{#6}
}
\makeatother

\begin{document}
\begin{frontmatter}

\title{Active Control and Sustained Oscillations \\ in  actSIS Epidemic Dynamics\thanksref{footnoteinfo}} 

\thanks[footnoteinfo]{The research was supported in part by Army Research Office grant W911NF-18-1-0325, Office of Naval Research grant N00014-19-1-2556, National Science Foundation grants CMMI-1635056, CNS-2027908, CCF-1917819, and C3.ai Digital Transformation Institute.}

\author[First]{Yunxiu Zhou} 
\author[Second]{Simon A. Levin} 
\author[Third]{Naomi Ehrich Leonard}

\address[First]{Operations Research and Financial Engineering, Princeton University, 
   Princeton, NJ 08544 USA (e-mail: yunxiuz@princeton.edu).}
\address[Second]{Ecology and Evolutionary Biology, Princeton University, 
   Princeton, NJ 08544 USA (e-mail: slevin@princeton.edu)}
\address[Third]{Mechanical and Aerospace Engineering, Princeton University, 
   Princeton, NJ 08544 USA (e-mail: naomi@princeton.edu)}

\begin{abstract}                
An actively controlled Susceptible-Infected-Susceptible (actSIS) contagion model is presented for studying epidemic dynamics with continuous-time feedback control of infection rates. Our work is inspired by the observation that epidemics  can be controlled through  decentralized disease-control strategies such as quarantining, sheltering in place, social distancing, etc., where individuals actively modify their contact rates with others in response to observations of infection levels in the population. Accounting for a time lag in observations and categorizing individuals into distinct sub-populations based on their risk profiles, we show that the actSIS model manifests qualitatively different features as compared with the SIS model.  In a homogeneous population of risk-averters, the endemic equilibrium is always reduced, although the transient infection level can exhibit overshoot or undershoot. In a homogeneous population of risk-tolerating individuals, the system exhibits bistability, which can also lead to reduced infection. For a heterogeneous population comprised  of risk-tolerators and risk-averters, we prove conditions on model parameters for the existence of a Hopf bifurcation and sustained oscillations in the infected population. 
\end{abstract}


\end{frontmatter}

\section{Introduction}
\label{sec:intro}


Deterministic compartmental models in epidemiology have provided valuable insights  for the understanding of evolutionary dynamics of infectious disease spread in a host population (\cite{kermack1927contribution}). These models have also been widely applied to study
various other spreading dynamics, including but not limited to information dissemination in social networks (\cite{jin2013epidemiological}), sentiment contagion in human society (\cite{zhao2014sentiment}) and
 the propagation of systemic risks in financial market (\cite{demiris2014epidemicfinance}).
 Although these deterministic models unavoidably ignore some important details such as individual heterogeneity and network structure, they adequately  capture the qualitative features of the spreading dynamics, including transient system behavior and stability of solutions. 
The models have been shown to be close approximations of certain Markov chain models of the underlying stochastic dynamics,  see, e.g., \cite{sahneh2013generalized}.
The Susceptible-Infected-Susceptible (SIS) model has been widely studied and applied in epidemiological modeling.
 While its assumption that individuals acquire  no immunity after recovery  may not be suitable for certain diseases, it provides a worst case scenario, which is valuable   for a large class of contagious diseases in general. 
 
The SIS model in its simplest form assumes a constant infection rate. However, it is well acknowledged that the rate can vary over time due to the different control strategies taken by individuals, which in turn affects the contagion dynamics. 
Indeed, understanding such models and their extensions is fundamental to developing disease control strategies. A review of analysis and control of epidemics is provided in \cite{nowzari2016analysiscontrol}. 


In this paper, we propose a model with feedback controlled infection rates to account for  active control strategies. 
Individuals modify their contact rates with others based on what they observe about the level of infection in the population. 
We model a time lag in the observations and distinguish individuals as \textit{risk-averters}, \textit{risk-tolerators},  and  \textit{risk-ignorers}. 
\textit{Risk-averters} represent those who change their contact rate with others in the opposite direction as the change in the observed infected population level, e.g., those who could and did stay at home and practiced increased social distancing during the COVID-19 pandemic as they saw the infected population grow.  \textit{Risk-tolerators} represent those whose contact rate with others changes in the same direction as the change in the observed infected population level, e.g.,  
health care workers, delivery workers, and other essential workers, who were obliged to work during the COVID-19 pandemic. 
\textit{Risk-ignorers}  
represent those who do not actively modify their contact rates. 

 Our work contributes to the literature in the following ways. 
First, while active and passive spreading was first distinguished in the social economics literature (see \cite{hartmann2008modeling}), our model rigorously demonstrates the differences between them, and we prove new results on the dynamics of contagion with active control. 
Further, our model 
serves as one form of the state-dependent approaches discussed in \citet{rands2010groupstate} that offer evolutionarily grounded ways for studying social contagion in  collective processes.
Second, our feedback model uses a low-pass filter of the measured infected population level. This models the observation delay and introduces an important robustness to uncertainty.  This is in contrast to contagion models which feed back infected population level directly, as in \cite{Baker2020,Franco2020}.
Third, we prove a new and relatively simple scenario under which sustained oscillations appear within epidemiological frameworks without external forcing. Understanding mechanisms that can lead to oscillations is critically important in the context of infectious disease spread
(\cite{LinAndreasenLevin1999,dushoff2004dynamicaloscilla,camacho2011explainingmultiwave}, \cite{xu2020mechanisticmultiwave}). It is likewise of great interest  in many other socio-economic processes,  
e,g., the rise and fall of business cycles (\cite{mishchenko2014oscillationsbusiness}) and fluctuation of 
behavioral preferences in social networks (\cite{pais2012hopfDareenGrammer}).
Fourth, in contrast to control strategies that require global knowledge about the exact underlying spreading dynamics, e.g., \cite{nowzari2016analysiscontrol}, our work provides evidence for the promise of tunable decentralized active control strategies to manage the dynamics of epidemics.


    The paper is organized as follows. In Section \ref{sec:background}, we review the 
    Susceptible-Infected-Susceptible  (SIS) model. 
    We introduce the actively controlled Susceptible-Infected-Susceptible (actSIS) model in homogeneous populations of \textit{risk-tolerators} and \textit{risk-averters} respectively in Section \ref{sec:homo}, where we analyze the equilibrium solutions and stability conditions and show their qualitatively different features as compared with the SIS model. 
    We examine the actSIS model in a heterogeneous population comprised of \textit{risk-tolerators} and \textit{risk-averters} in Section \ref{sec:hetero}, and prove conditions for a Hopf bifurcation with a stable limit cycle. We conclude 
    in Section \ref{sec:disc}.

\section{Background}
\label{sec:background}


\subsection{SIS in a well-mixed population}

Consider a disease spreading in a large, randomly-mixed population, where  individuals are divided into either susceptible (S) or infected (I) classes. 
Susceptible individuals get infected at rate $\beta$ while infected individuals recover at rate $\delta$. 
Let $p(t)$ be the fraction of infected individuals and $s(t)$ the fraction of susceptible individuals at time $t$.  The SIS model is 
\begin{equation}
    \dot s =- \beta sp + \delta p, \;\;\; 
\dot p = \beta sp  - \delta p. 
     \label{SIS}
\end{equation}

The force of infection term $\beta p$ in \eqref{SIS} describes the rate at which susceptible individuals get infected (\cite{muench1934derivationforceofinfection}). It can be decomposed into three terms: $\beta p = \Bar{\beta} \times \alpha \times p$, where $\Bar{\beta}$ is the transmission rate of the disease and an intrinsic property of the disease, and $\alpha$ is the effective number of contacts per unit time. 
Since $s(t) = 1 - p(t)$, \eqref{SIS} can be rewritten as
\begin{equation}
    \begin{aligned}
    \dot{p}=\beta \left(1-p_{}\right) p -\delta p_{}.
\end{aligned}
\label{ode-sis-p}
\end{equation}

The steady-state behavior of solutions to the well-mixed SIS model \eqref{ode-sis-p} is characterized by the basic reproduction number $\mathcal{R}_{0} = \beta/\delta$, a key concept in epidemiology that defines the epidemic threshold of a particular infection (\cite{diekmann1990definitionR0}). 
If $\mathcal{R}_{0}>1$,  the disease persists and a nonzero fraction of the population is  infected at steady state: as $t \rightarrow \infty$, $p(t) \rightarrow 1-\delta/\beta$, the endemic equilibrium (EE).  If $\mathcal{R}_{0} \leq 1$, the disease dies out at steady state: as $t \rightarrow \infty$, $p(t) \rightarrow 0$, the infection free equilibrium (IFE). 
At $\mathcal{R}_{0}=1$, 
there is a transcritical bifurcation. 

\subsection{Network SIS model}
A natural extension of the homogeneous population setting is the introduction of heterogeneities of various kinds. Characterizing population heterogeneity in terms of infection and/or recovery rates has been discussed both in population subgroups (\cite{anderson1992infectious}) and in networks 
(\cite{hethcote1984lecture}, \cite{renato2019adaptive}). In the following, we review the network SIS model that was originally introduced as the multi-group SIS model in \cite{lajmanovich1976deterministicnetworkSIS}. 

Consider a heterogeneous population of $n$ sub-populations, each large, well-mixed and homogeneous.  
Let susceptible individuals in sub-population $i$ get infected through  contact with infected individuals in sub-population $j$ at rate $\beta_{i j} \geq 0$ and infected individuals in sub-population $i$ recover at rate $\delta_i\geq 0$.  The rate $\beta_{ij} $ can be decomposed as $\beta_{ij} = \bar\beta\times \alpha_{ij}$ where $\bar\beta$ is the transmission rate and $\alpha_{ij}$ represents the effective contact rate between sub-population $i$ and $j$.
The network SIS model is
\begin{equation}
    \begin{aligned}
    \dot{p_i}= \left(1-p_{i}\right) \sum_{i=1}^{n}\beta_{ij} p_{j} -\delta_{i} p_{i},
\end{aligned}
\label{ntwork-sis-p}
\end{equation}
where  $p_i(t)$ denotes the  fraction of infected individuals in the $i$th sub-population, or equivalently, the probability that a typical individual in sub-population $i$ is infected at time $t$. 
Let $B=\left\{\beta_{j k}\right\}$ and $\Gamma=\operatorname{diag}\left(\delta_{1}, \ldots, \delta_{N}\right)$ be the
infection matrix and the recovery matrix, respectively. 

For the network SIS model \eqref{ntwork-sis-p}, the basic reproduction number is $\mathcal{R}_{0}=\rho\left(B \Gamma^{-1}\right)$, where $\rho$ denotes the spectral radius. For $\mathcal{R}_{0} \leq 1,$ solutions converge to the IFE as $t \rightarrow \infty$ while for $R_{0}>1,$ solutions converge to the EE. See \cite{lajmanovich1976deterministicnetworkSIS}, \cite{fall2007epidemiologicalnetworkSIS}, \cite{BulloEpidemics2018}  for details. 

\section{Homogeneous Population}
\label{sec:homo}
 Our proposal of the actively controlled Susceptible-Infected-Susceptible (actSIS) model is based on the following observations. 
 First, individuals conduct behavioral changes  as they acquire information of the epidemics which consequently affect their contact rates with others (\cite{funk2009spreadawareness}). 
 Such information, however, often involves estimations that are delayed  
 and  unavoidably omits details in the finer time scale. We are motivated in part by the model and study of change in susceptibility after first infection as presented \cite{renato2019adaptive}. 
 
 Accordingly, we present the actSIS model for studying epidemic dynamics with continuous-time feedback control of infection rates. 
 We let the infection rate be
the product of the intrinsic transmission rate $\bar{\beta}$ and an effective contact rate $\alpha(\cdot)$ that is actively modified by individuals based on their observations of the system state. To account for the uncertainty and delay in measurements of the infection level in the population, we let the feedback responses depend on the filtered state $p_s$ of the infected fraction $p$, where $p_{s}$ tracks 
$p$ and possibly some external stimulus $r(t)$ with a time constant $\tau_{s}$. The actSIS model is
\begin{equation}
    \begin{aligned}
    \dot{p} &=  \beta^{}  (1-p) p - \delta p,\\
     \tau_s \dot{p_{s}} &= -p_{s} + p + r, \\
    \beta^{} &=\Bar{\beta}\alpha^{}(p_{s}). 
    \end{aligned}
    \label{orginal-ASIS-dynamics}
\end{equation}

Acknowledging that people conduct social-behavioral changes in a soft-threshold manner (\cite{smaldino2018sigmoidal}), we consider sigmoidal-shaped functions for the feedback response $\alpha(\cdot)$. Similar feedback mechanisms  in neuronal dynamics have been shown to exhibit ultra-sensitivity and robustness to inputs and variability (\cite{sepulchre2019activecontrol}).
Let $\phi: [0,1] \rightarrow [0,1]$ be a monotonically increasing saturating function 
\begin{equation*}
    \phi(p;\mu,\nu)= \left(1+\left(\frac{p\cdot(1-\mu)}{\mu\cdot(1-p)}\right)^{-\nu}\right)^{-1},
\end{equation*}
with location parameter $\mu\in(0,1)$ and slope parameter $\nu\in(0,1)$.\footnote{This particular form of a sigmoidal function over the unit interval was proposed by \cite{antweiler_2018}. $\mu$ controls the value of $p$ at which $\Phi(p)=1/2$ and $\nu$ controls how gradually or sharply the function grows.}
For \textit{risk-tolerators} we define $\alpha$ to vary directly with $p_s$:
\begin{equation}
\begin{aligned}
\alpha^{}\left(p_{s}\right)&= \phi\left(p_{s} ; \mu_{T}, \gamma_{T}\right) =: \phi^{T}\left(p_{s}\right).\\
\end{aligned}
\end{equation}
For \textit{risk-averters} we define $\alpha$ to vary inversely with $p_s$:
\begin{equation}
\begin{aligned}
\alpha^{}\left(p_{s}\right)&=1-\phi\left(p_{s} ; \mu_{A}, \gamma_{A}\right)=: 1-\phi^{A}\left(p_{s}\right).
\end{aligned}
\end{equation}


It is not surprising to find that the EE of the actSIS model for both \textit{risk-tolerators} and \textit{risk-averters} is upper bounded by that of the SIS model, since the incorporation of feedback responses $\alpha(\cdot) \in [0,1]$ always decreases the effective infection rates. However, 
the underlying structure resulting in these lower endemic solutions differs among types of individuals, and
 the actSIS model shows  qualitatively different dynamical features as compared to the SIS model. 
 
 For homogeneous \textit{risk-tolerators},  \eqref{orginal-ASIS-dynamics} undergoes a saddle node bifurcation and exhibits bistability as illustrated in Fig.~\ref{fig:bistable} and Fig.~\ref{fig:bifur}(a).  The saddle node bifurcation point is greater than the transcritical bifurcation point in the SIS model, which has implications for control design since  
 it implies that it is more difficult for the disease to spread in a population of \textit{risk-averters} than in a population of \textit{risk-ignorers}. 

For  homogeneous \textit{risk-averters}, \eqref{orginal-ASIS-dynamics} undergoes a transcritical bifurcation, as in the SIS model, but with a reduced EE, as illustrated in Fig.~\ref{fig:bifur}(b). Further, the EE becomes a stable focus under certain conditions,  resulting in large overshoot and/or undershoot in the transient dynamics, as illustrated in Fig.~\ref{fig:focus}. 

We devote the rest of this section to the detailed description and proof of these rich dynamics. 
For simplicity of exposition, we set $\delta=1, \tau_s=10$ throughout the rest of this paper.

\begin{figure}[!htbp]
\centering
\subfigure[IFE]{
    \includegraphics[width=0.2\textwidth]{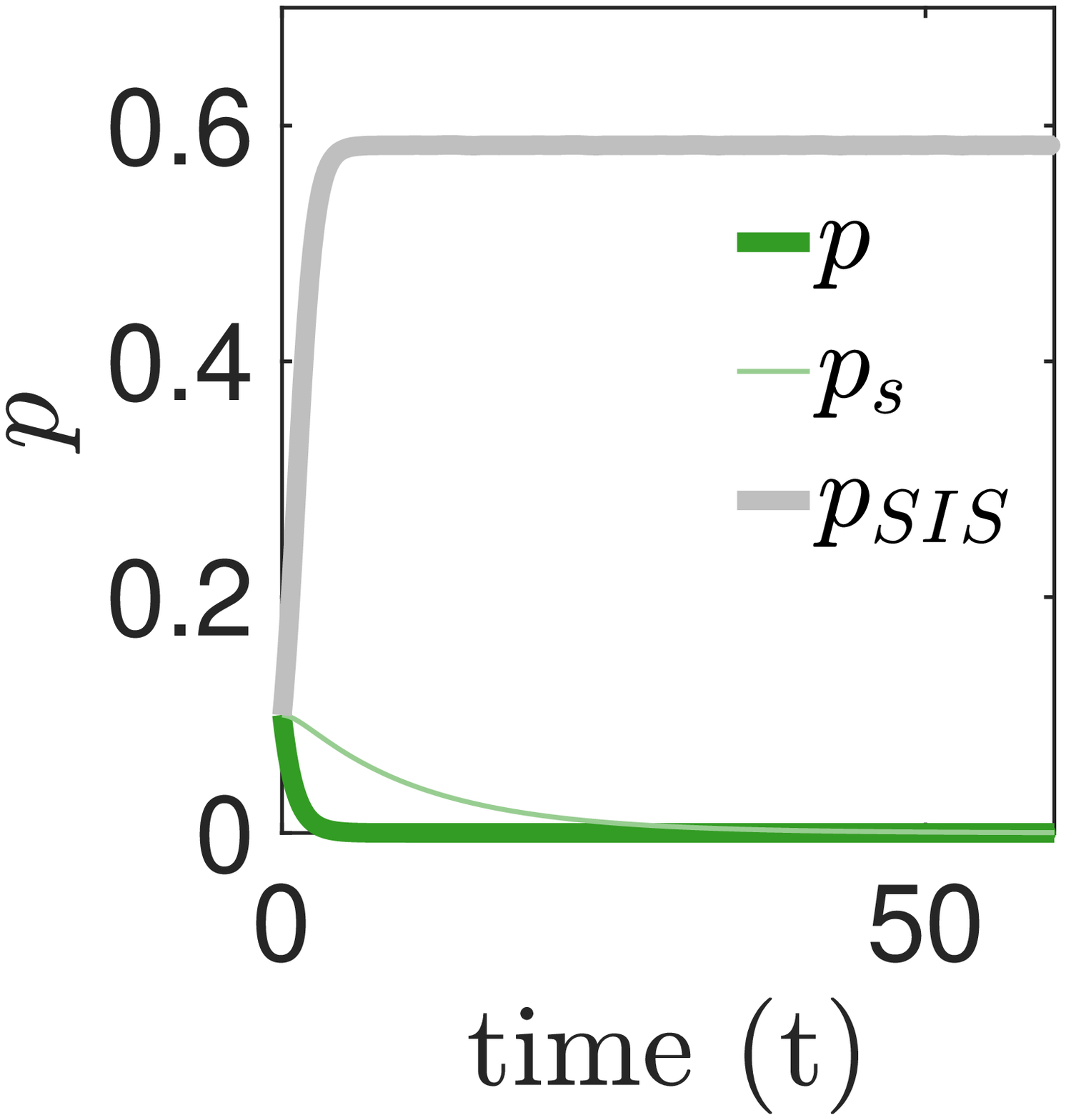}
}
\subfigure[EE]{
    \includegraphics[width=0.2\textwidth]{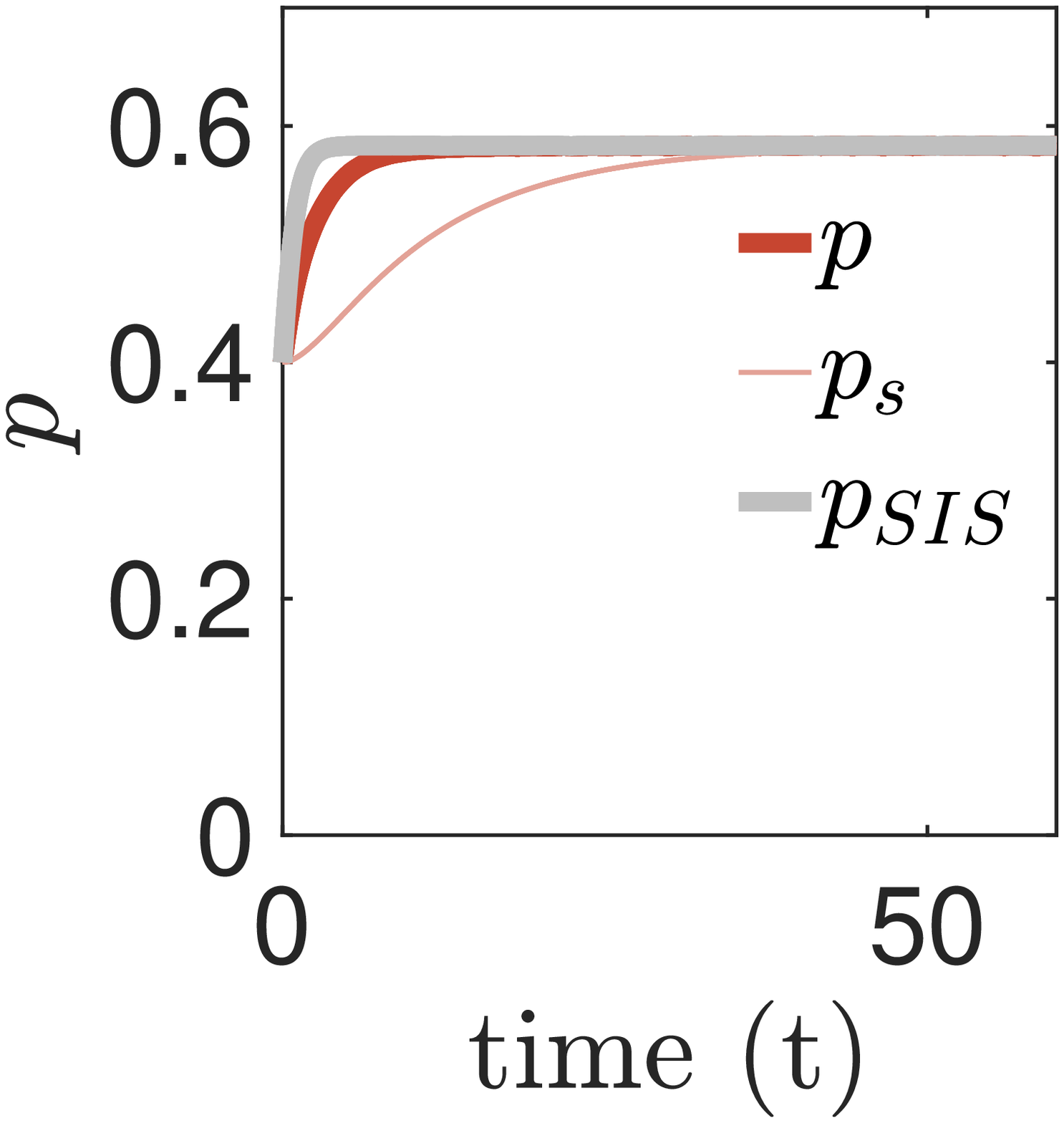}
}
  \centering
\figurecaption{%
Bi-stability of \textit{risk-tolerators}%
}{%
The IFE and EE are both stable solutions for $\bar\beta > \bar\beta_c$ in the actSIS model for \textit{risk-tolerators}. a) For 
$p(0)=0.1$, the solution (green) of actSIS (\textit{risk-tolerators}) converges to the IFE, whereas the solution (grey) of SIS (\textit{risk-ignorers}) converges to the EE. b) For 
$p(0)=0.4$, the solution (red) of actSIS and the solution (grey) of SIS converge to the EE.  
Parameters for all simulations are $\bar{\beta} =2.4, \mu_T=0.35,\nu_T=8.$ For these parameters, $\bar\beta_c = 1.87< \bar\beta$.
}
\label{fig:bistable}
\end{figure}

\begin{figure}[!htbp]
\centering
\subfigure[Risk-tolerators]{
    \includegraphics[width=0.2\textwidth]{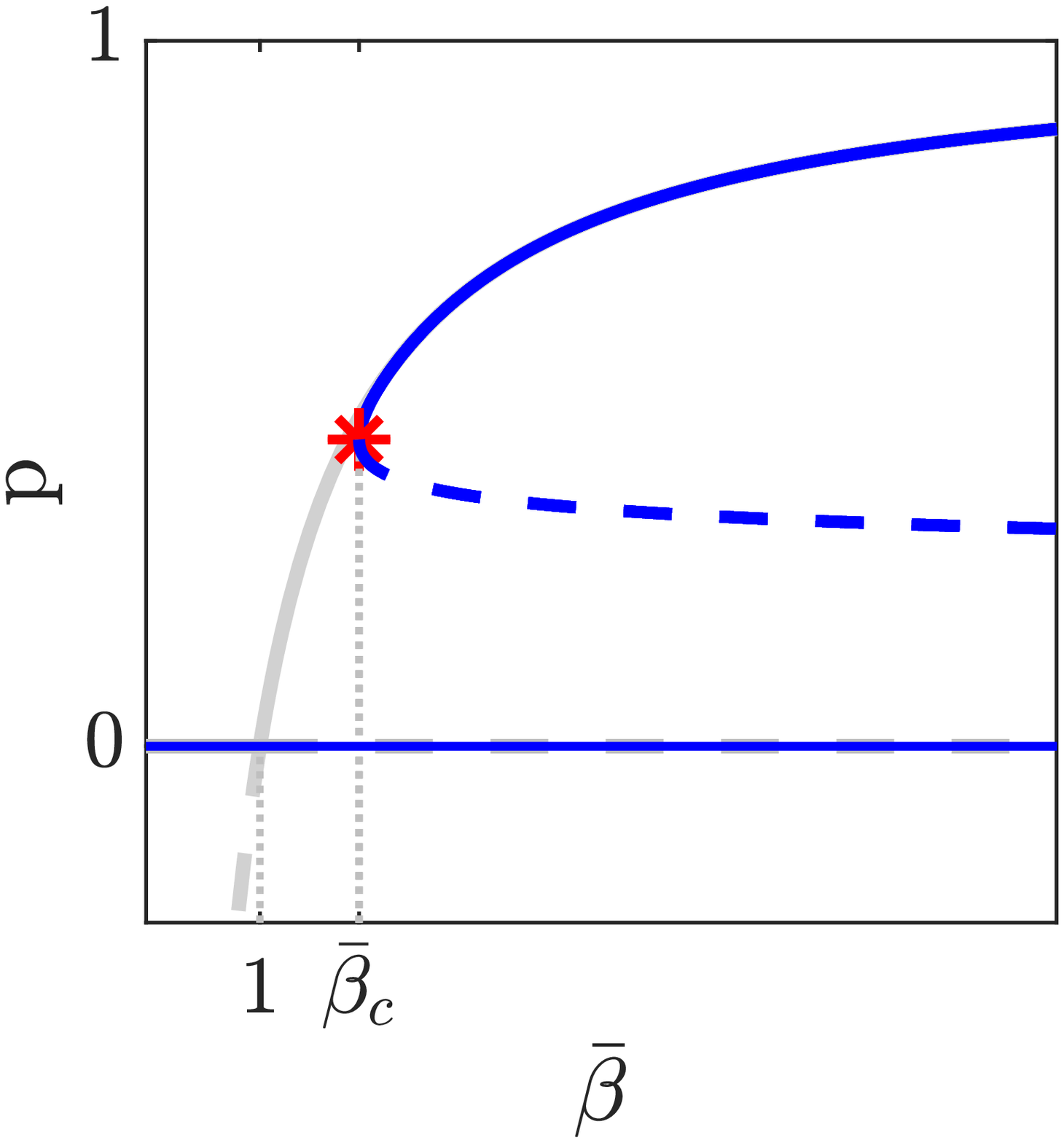}
}
\subfigure[Risk-averters]{
    \includegraphics[width=0.2\textwidth]{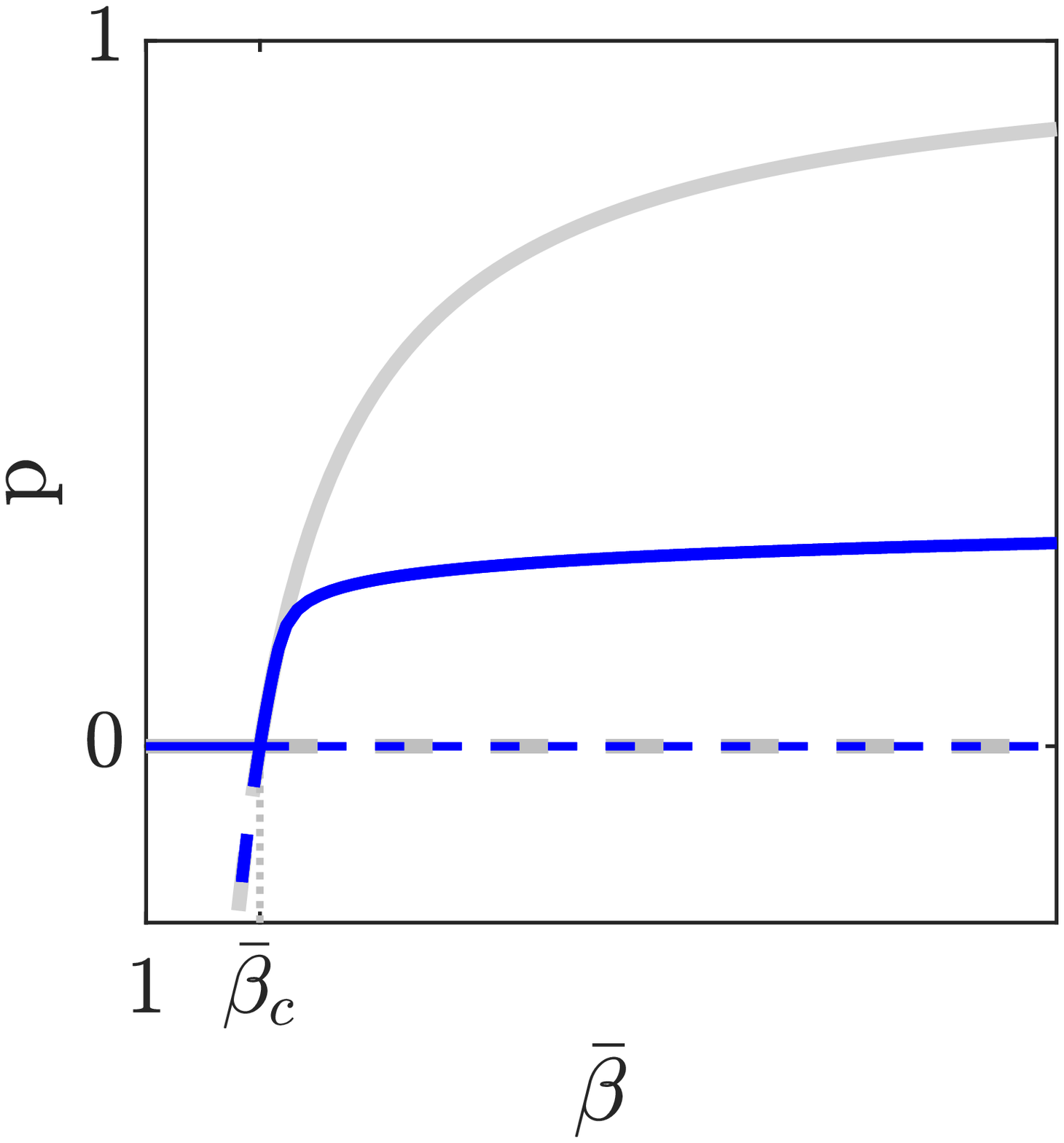}
}
  \centering
\figurecaption{%
Bifurcation diagrams of the actSIS model%
}{%
a) For homogeneous \textit{risk-tolerators}, the actSIS model  undergoes a saddle-node bifurcation  at  bifurcation point $\bar{\beta}_c>1$. b) For homogeneous \text{risk-averters}, the actSIS model undergoes a transcritical bifurcation at $\bar\beta=1$ as in the SIS model. The EE for the actSIS models are upper-bounded by that of the SIS model in both cases. 
The bifurcation diagrams of the actSIS models are drawn in blue and the SIS in grey. Solid curves denote stable solutions and dashed curves denote  unstable solutions. Parameters for these plots are $ \mu_T=0.35,\nu_T=8, \mu_A=0.25,\nu_A=8.$ The bifurcation diagrams: Fig.1 and Fig. \ref{fig:bifur-Hopf} (a), are produced with the MATLAB package matcont (\cite{dhooge2003matcont}).

}
\label{fig:bifur}
\end{figure}

\vspace{1.8mm}

\begin{theorem} [Steady-state behavior of the actSIS model with homogeneous risk-tolerators] Consider the actSIS dynamics  for a homogeneous population of \textit{risk-tolerators} given by \eqref{orginal-ASIS-dynamics} with $r(t) = 0$ and $\alpha(\cdot) =\phi^{T}(\cdot)$.  Then the following hold:
\begin{enumerate}[label=(\roman*)]
    \item There are  two  types of equilibrium solutions:
\begin{enumerate}
    \item IFE: $p=p_{s}=0$. It is always stable; 
    \item EE: $p=p_{s}=p^*$ satisfying  $\phi^{T}(p^*)(1-p^*) = \frac{\delta}{\Bar{\beta}}$. There are zero, one, or two such solutions. A solution is stable if $\Phi^{T \prime}(p^*) \cdot(1-p^*)^{2}<\frac{\delta}{\bar{\beta}}$. 
\end{enumerate}

\item The system undergoes a saddle node bifurcation at the bifurcation point $(\Bar{\beta}_c, p^*)$ with stable upper branch (stable EE) and unstable lower branch (unstable EE). The epidemic threshold $\Bar{\beta}_c/\delta$ is always larger than that of the SIS model, i.e., $\Bar{\beta}_c/\delta>1$.

\item For $\Bar{\beta}>\Bar{\beta}_c$ the system exhibits bistability of the IFE and the EE. For $\Bar{\beta}\leq\Bar{\beta}_c$ the IFE is the only stable equilibrium.

\item  As  $\Bar{\beta}_c<\Bar{\beta} \rightarrow \infty$, 
$p^* \rightarrow 1-\delta/\Bar{\beta}$, the EE of SIS.  

\end{enumerate}
\end{theorem}

\begin{proof} 
\begin{enumerate}[label=(\roman*)]
    \item The equilibrium solutions are straightforward to compute. For stability, we compute the Jacobian  as 
\begin{equation}
    J^T = \begin{bmatrix}
                     \Bar{\beta} \phi^T (p_{s}) (1-2 p)-\delta & \Bar{\beta} (1-p) p {\phi^{T}}^{\prime} (p_{s}) \\
                    \frac{1}{\tau_s}  & -\frac{1}{\tau_s} 
                    \end{bmatrix}.
                    \label{gradient}
\end{equation}
For the IFE,
\eqref{gradient} reduces to 
\begin{equation}
    \left.J^T\right|_{ \text{IFE}}= \begin{bmatrix}
                     -\delta & 0 \\
                    \frac{1}{\tau_s}  & -\frac{1}{\tau_s} 
                    \end{bmatrix},
\end{equation}
implying that the IFE is always stable. For the EE,
\begin{equation}
    \left.J^T\right|_{ \text{EE}}= \begin{bmatrix}
                     -\delta \frac{p^*}{1-p^*} & \Bar{\beta} (1-p^*) p^* {\phi^{T}}^{\prime}(p^*)  \\
                    \frac{1}{\tau_s}  & -\frac{1}{\tau_s} 
                    \end{bmatrix}.
\end{equation}
Stability requires the equilibrium solution to satisfy  $\frac{\delta}{1-p^*} - \Bar{\beta} (1-p^*)\phi^{T \prime}(p^*) >0 $. Since $p^* \in (0,1)$ , it is equivalent to requiring that  $ \phi^{T \prime}(p^*)\cdot (1-p^*)^2 < \frac{\delta}{\Bar{\beta}}$. 

\item 
We start by computing the critical value $\Bar{\beta}_c$. Since $\phi^T(\cdot)$ is a monotonically increasing function taking values between 0 and 1, $g(p):=\phi^T(p)(1-p)$ takes value between 0 and 1 and it first increases from 0 (since $g(0)=0$) and then decreases to 0 ($g(1)=0$). Depending on parameter values ($\mu_T, \nu_T, \Bar{\beta},\delta$), the EE  has either zero, one or two solutions. Let $h(p):= \phi^{T \prime}(p)\cdot (1-p)^2 $. It first  increases and then decreases for $p\in [0,1]$. Denoting $\hat{p}:= \argminF_p g(p) $, one can check that $h(p)$ intersects with $g(p)$ at three points: $0, \hat{p}$ and $1$. This implies that when the EE has two solutions (when $\delta/\beta < g(\hat{p})$), the smaller solution is unstable and the larger one is stable.
$\Bar{\beta}_c$ can be solved analytically by  solving $\Bar{\beta}_c=\delta/g(\hat{p})$. 

To prove the existence of a saddle-node bifurcation, we use the classification of equilibria for a two-dimensional system 
presented in Section 4.2.5 in \cite{izhikevich2007dynamicalnote2}. Specifically, we show that one of the eigenvalues of the Jacobian  at $(\Bar{\beta}_c, p^*)$ becomes zero. Using equations $\Bar{\beta}_c=\delta/g(\hat{p})$ and $g^{\prime}(\hat{p})=0$, one can easily verify the determinant of the Jacobian  at $(\Bar{\beta}_c, p^*)$ equals  zero thus the existence of a saddle node bifurcation. 

Since $g(\hat{p})\in (0,1)$, we have $ \Bar{\beta}_c/\delta$ is always larger than the epidemic threshold in the SIS model. 

\item This follows directly from (\romannumeral 1) and (\romannumeral 2).

\item Comparing the equations satisfied by the EE solutions for the SIS and the homogeneous actSIS (\textit{risk-tolerators}) model, we observe that $\phi^T(p^*)$ and $p^*$ increase with $ \Bar{\beta}$. As $\phi^T(p^*)$ approaches  $1$, $p^*$ approaches  $1-\delta/\Bar{\beta}$. 
\end{enumerate}
\end{proof}

\begin{figure}[!htbp]
\centering
\subfigure[Phase Portrait]{
    \includegraphics[width=0.21\textwidth]{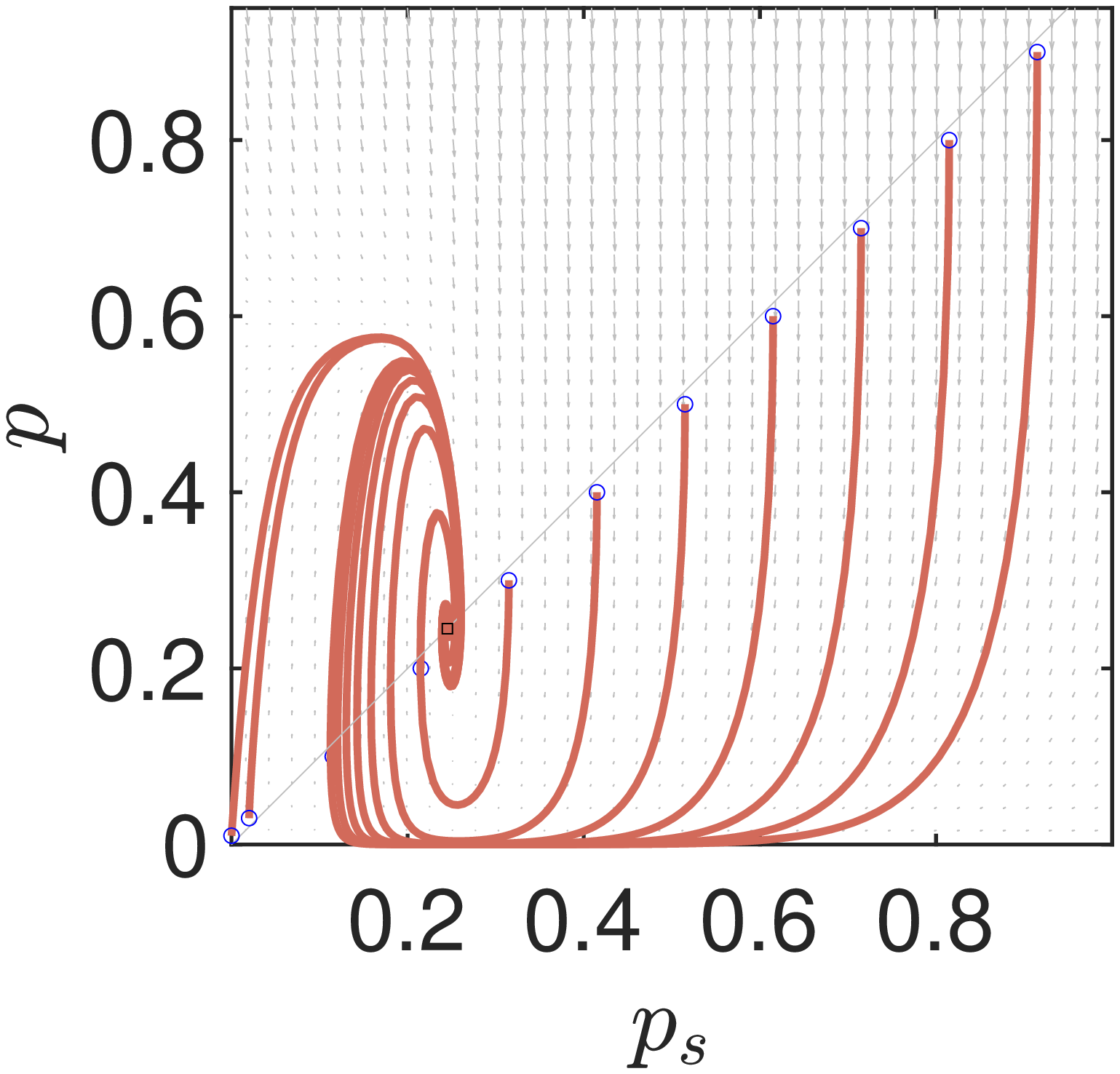}
}
\subfigure[Undershoot and overshoot]{
    \includegraphics[width=0.2\textwidth]{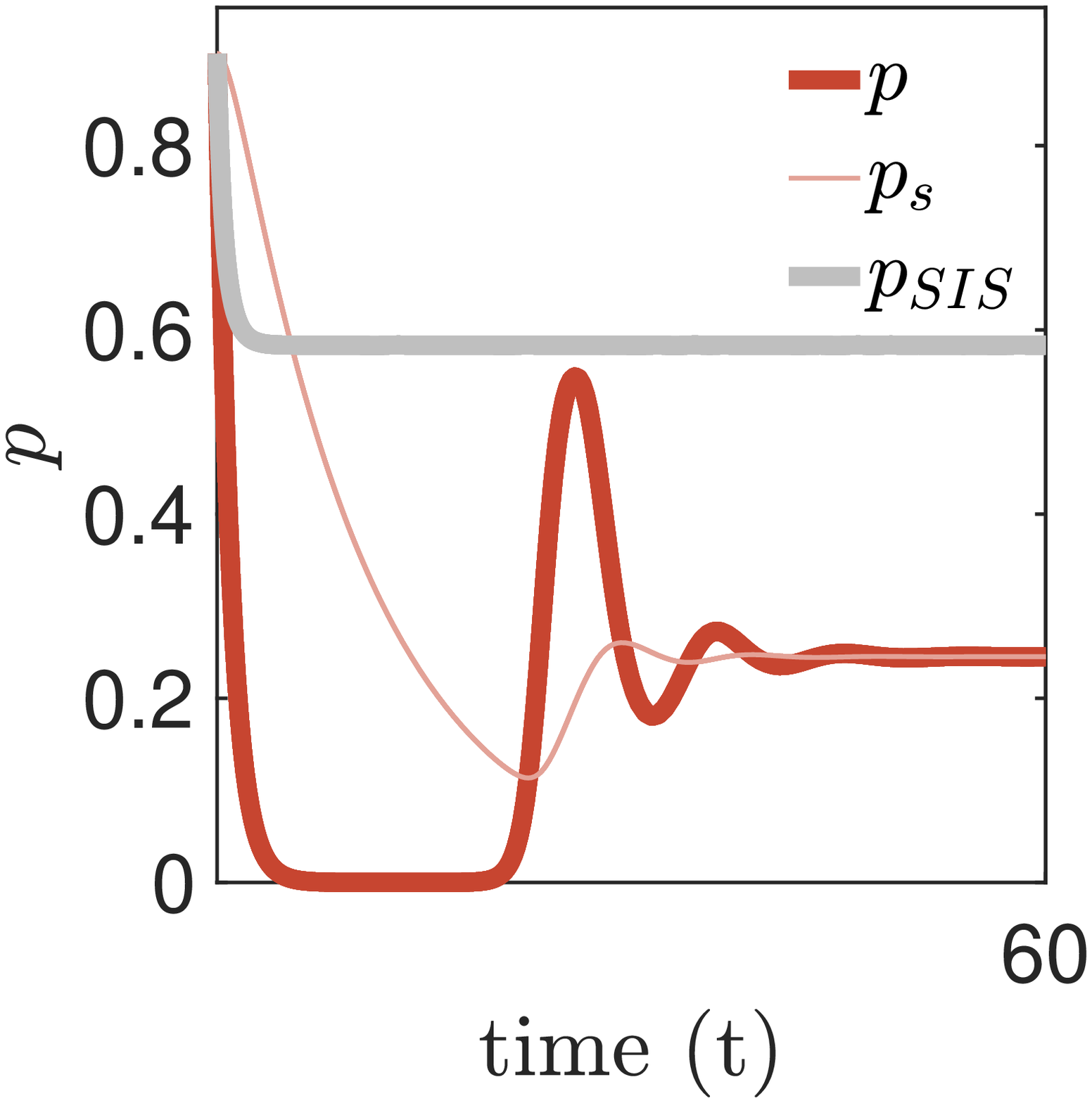}
}

  \centering
\figurecaption{%
The EE as a stable focus%
}{%
(a) Phase portrait of the actSIS model for a homogeneous population of \textit{risk-averters}. Grey arrows depict the vector fields. The initial conditions and end points of the simulations are plotted as circles and squares respectively. Starting from low initial values, $p$ exhibits a rapid increase followed by a decrease to the EE. Starting from high initial values, $p$ exhibits a rapid decrease to nearly zero  before an increase to the equilibrium state.
(b) For $p(0) =p_s(0) = 0.9$, the transient state $p$ undershoots to near zero for a long time and then overshoots. The parameters  are $\bar{\beta}=2.4,\mu_A=0.25, \nu_A=8$. 
}
\label{fig:focus}
\end{figure}

\begin{theorem} [Steady-state behavior of the actSIS model with homogeneous risk-averters] Consider the actSIS dynamics  for a homogeneous population of \textit{risk-averters} given by \eqref{orginal-ASIS-dynamics} with $r(t) = 0$ and $\alpha(\cdot) =1-\phi^{A}(\cdot)$. Then the following hold:
\begin{enumerate}[label=(\roman*)]
    \item There are  two equilibrium solutions:
\begin{enumerate}
    \item IFE: $p_s = p = 0$. It is stable if $  \Bar{\beta} < \delta$; 
    \item EE: $p_s = p = p^*$  satisfying  $  (1 - \Phi^A(p^*))(1-p^*) = \delta/\Bar{\beta}$. It is always stable if it exists, which is when $\Bar{\beta} > \delta$.
\end{enumerate}

\item The EE 
is a stable focus if $b > (a-c)^2/(4ac)$,
where  $a=\frac{\delta p^*}{1-p^*}, b=\frac{1-p^*}{1-\phi^A(p^*)}\phi^{A \prime} (p^*), c=\frac{1}{\tau_s}$.

\item  
The EE is always upper bounded by $1-\delta/\Bar{\beta}$, the SIS EE. 

\end{enumerate}
\end{theorem}

\begin{proof}
\begin{enumerate}[label=(\roman*)]
    \item Since $1-\Phi^A(\cdot)$ is monotone decreasing, there exists at most one  EE solution when $\delta/\bar{\beta} <1$. 
The Jacobian  of \eqref{orginal-ASIS-dynamics} is 
\begin{equation}
   J = \begin{bmatrix}
                     \Bar{\beta} (1-\Phi^A(p_{s})) (1-2 p)-\delta & - \Bar{\beta} (1-p) p \phi^{A \prime}(p_{s}) \\
                    \frac{1}{\tau_s}  & -\frac{1}{\tau_s} 
                    \end{bmatrix},
                    \label{gradientASIS2}
\end{equation}
which simplifies to 
 \begin{equation}
    \left.J\right|_{ \text{IFE}}= \begin{bmatrix}
                     \Bar{\beta} -\delta & 0 \\
                    \frac{1}{\tau_s}  & -\frac{1}{\tau_s}
                    \end{bmatrix},
\end{equation}
for the IFE, and
 \begin{equation}
    \left.J\right|_{ \text{EE}}= \begin{bmatrix}
                     - \frac{\delta p^*}{1- p^*} & - \Bar{\beta} (1-p^*) p^* \phi^{A \prime}(p^*) \\
                    \frac{1}{\tau_s}  & -\frac{1}{\tau_s} 
                    \end{bmatrix}
                    \label{eq:jaco-ee}
\end{equation}
for the EE. Therefore, the IFE is stable when $  \Bar{\beta} < \delta$. For the EE, stability requires $\frac{\delta p^*}{1-p^*} + \Bar{\beta} (1-p^*) p^* \Phi^{A\prime}(p^*) >0 $. Since ${\Phi^A}'(p^*)$ is always non-negative, the stability conditions always holds when such a solution exists. 

\item We prove when EE is  a stable focus, by proving when
     \eqref{eq:jaco-ee}  has a pair of complex-conjugate eigenvalues with negative real part. 
We use the classification of equilibria for a two-dimensional system according to the trace  and  determinant of the Jacobian  (\cite{izhikevich2007dynamicalnote2}). Denoting $A:=\left.J\right|_{ \text{EE}}$  and using the definitions of $a, b$ and $c$, we have $\Tr(A)=-(a+c)$ and $ \det(A)=ac(1+b)$. Therefore the condition $b>\frac{(a-c)^2}{4ac}$, guarantees that $\Tr(A)^2-4\det(A)<0$, and together with $\Tr(A)<0$, that the EE is a stable focus.

\item As $ \Bar{\beta}$ increases, the right hand side of  $  (1 - \Phi^A(p^*))(1-p^*) = \delta/\Bar{\beta}$ decreases. As a result, $p^*$ increases with $ \Bar{\beta}$. On the other hand, $1-\phi^A(p^*)$ decreases away from $1$ thus the EE are always bounded above by $1-\delta/\Bar{\beta}$.
\end{enumerate}

\end{proof}


\section{Heterogeneous Population}
\label{sec:hetero}
We introduce the heterogeneous network actSIS model.
The transmission rates, recovery rates and feedback responses  can all be  distinct. 
However, we restrict to the following set up as a first step in exploring the role of heterogeneity. 
  Let the population comprise two homogeneous sub-populations, \textit{risk-tolerators} and \textit{risk-averters}, which differ only in their feedback responses to the infection. 
   We assume the disease transmission occurs across sub-populations but not within. 
  For the generalization of the network SIS model \eqref{ntwork-sis-p} to the network actSIS model, this  translates as $n=2$, $\bar{\beta}_i=\bar{\beta}$, $\delta_i = \delta$,  $\beta_{ii} = 0$,  $\alpha(\cdot)=\alpha^1(\cdot)$ for all \textit{risk-tolerators} and $\alpha(\cdot)=\alpha^2(\cdot)$ for all \textit{risk-averters}.  
  Let $p_1(t)$ ($p_2(t)$) denote the fraction of \textit{risk-tolerators} (\textit{risk-averters}) that are infected at time $t$.  
  Let $\beta_{12}=\bar{\beta}\alpha^1(p_{s_{2}})=\bar{\beta} \phi^{T}\left(p_{s_{2}}\right)$ be the effective infection rate from  \textit{risk-averters} to \textit{risk-tolerators}, and $\beta_{21}=\bar{\beta}\alpha^2(p_{s_{1}}) =\bar{\beta} \left(1-\phi^{A}\left(p_{s_{1}}\right)\right)$  from  \textit{risk-tolerators} to  \textit{risk-averters}. 
The heterogeneous network actSIS model in the case of these two subpopulations is
\begin{equation}
\begin{aligned}
&\dot{p}_{1}=\left(1-p_{1}\right) \beta_{1 2} \cdot p_{2}-\delta p_{1},\\
&\dot{p}_{2}=\left(1-p_{2}\right) \beta_{2 1} \cdot p_{1}-\delta p_{2},\\
&\tau_s \dot{p}_{s_{1}}=-p_{s_{1}}+p_{1},\\
&\tau_s \dot{p}_{s_{2}}=-p_{s_{2}}+p_{2}.
\label{eq:two-node}
\end{aligned}
\end{equation}

The equilibrium solutions  
satisfy $p_{1}^{*}=p_{s_{1}}^{*}, p_{2}^{*}=p_{s_{2}}^{*}$ and 
\begin{equation}
    \frac{p_{1}^{*}}{1-p_{1}^{*}} \cdot \frac{1}{\alpha^{1}\left(p_{2}^{*}\right)p_{2}^{*}}=\frac{p_{2}^{* }}{1-p_{2}^{*}} \cdot \frac{1}{\alpha^{1}\left(p_{2}^{*}\right)p_{1}^{*}} =\frac{\delta}{\bar{\beta}}.
\end{equation}

As shown in Fig.~\ref{fig:bifur-Hopf}(a), there are model parameters for which  system \eqref{eq:two-node} undergoes a Hopf bifurcation with a stable limit cycle.  As  $\Bar{\beta}$ increases from zero, system \eqref{eq:two-node} undergoes a saddle-node bifurcation from a single stable IFE to bistability of  the EE and the IFE. As $\Bar{\beta}$ increases further across the Hopf bifurcation point  $\Bar{\beta}^*$, the system exhibits a stable limit cycle about the EE. 
 For completeness, we present the following theorem (Theorem 3.4.2 in \cite{guckenheimer2013nonlinearHopfProof}), which we use to prove the existence of stable limit cycles for \eqref{eq:two-node}. 

\begin{figure}[!htbp]
\centering
\subfigure[Hopf bifurcation of system \eqref{eq:two-node}]{
    \includegraphics[width=0.22\textwidth]{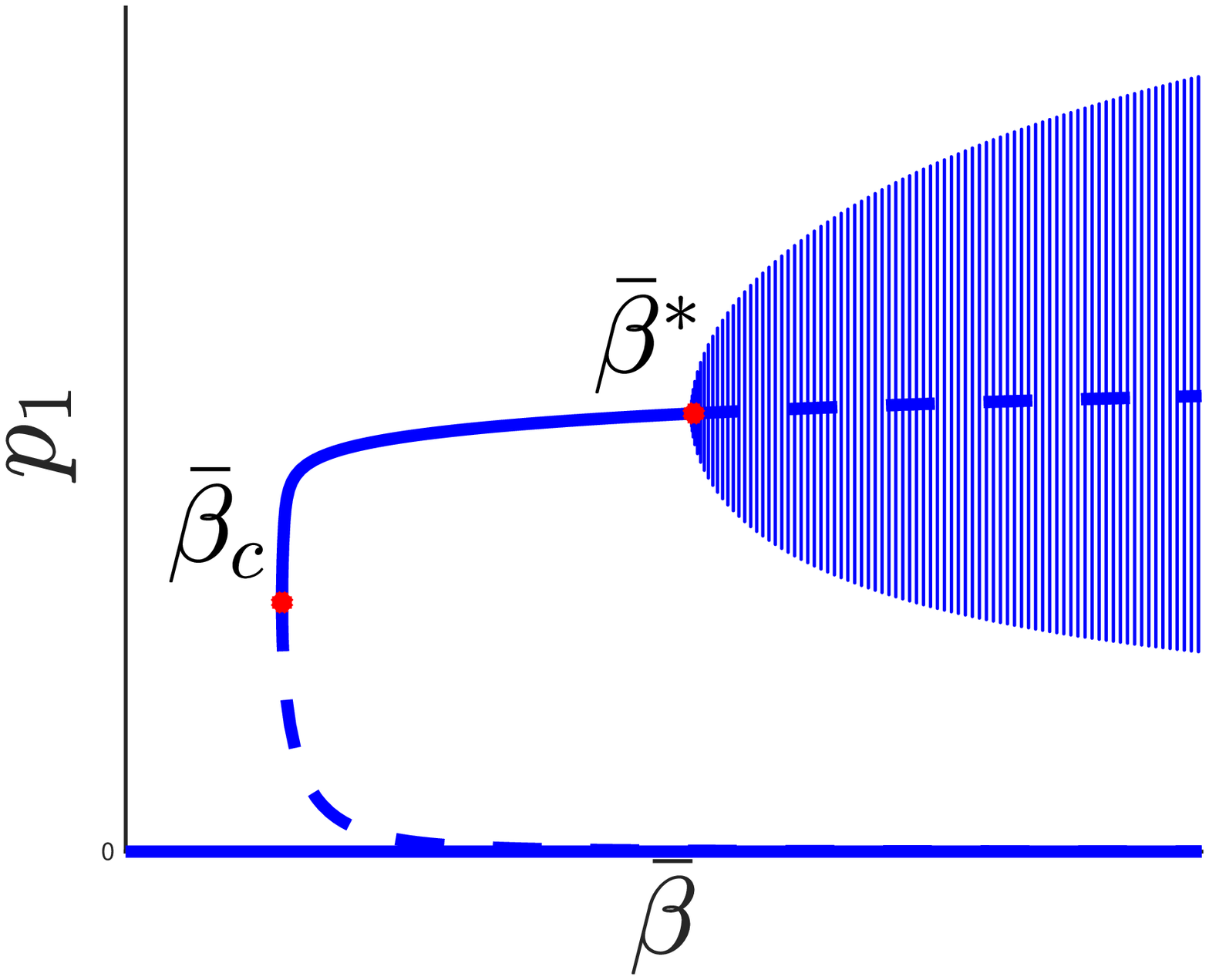}
}
\subfigure[Verification of (H2)]{
    \includegraphics[width=0.21\textwidth]{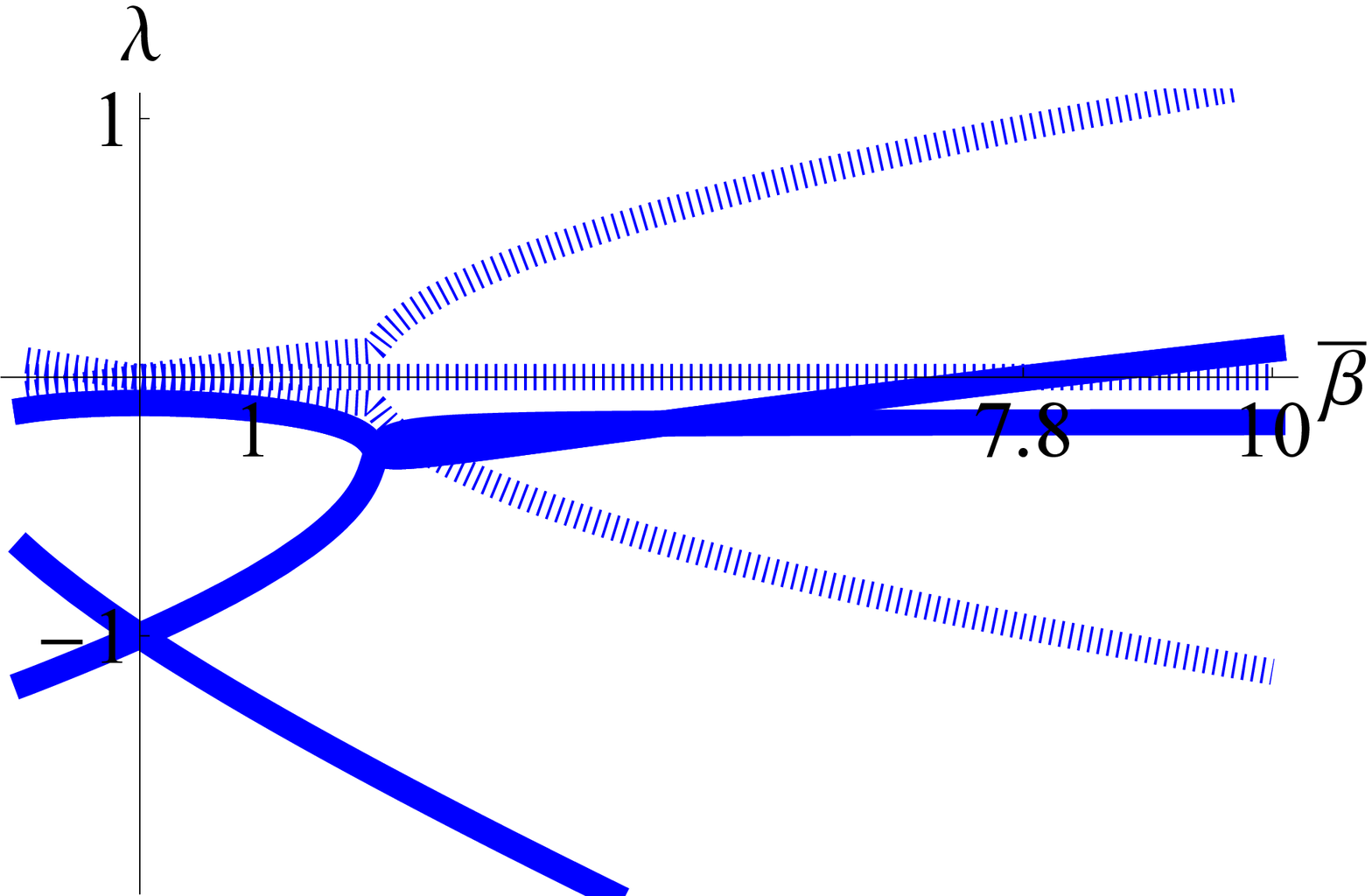}
}
  \centering
\figurecaption{%
Hopf bifurcation}{
(a) Bifurcation diagram for $p_1$ with bifurcation parameter $\bar\beta$. For small positive values of $\bar\beta$, the IFE is the only stable solution.  As $\bar\beta$ increases  a saddle-node bifurcation leads to bistability of the IFE and EE.  As $\bar\beta$ continues to increase there is a Hopf bifurcation (red dot) and stable oscillations about the EE.
The blue solid curves depict  stable solutions, while dashed curves depict unstable solutions. 
(b) The real and imaginary parts of the eigenvalues of the Jacobian $D_{p} \boldsymbol{f}$ at $(\boldsymbol{p}^{\mathbf{*}}, \Bar{\beta}^*)$ 
are plotted in solid and dashed lines. At around $\Bar{\beta}^*=7.8$, a pair of complex eigenvalues crosses the real line ($x$-axis) with a nonzero derivative, verifying (H2). Parameters are $\delta = 1, \tau_s = 10, \mu_T=0.45,\nu_T=0.8, \mu_A=0.45,\nu_A=20.$
}
\label{fig:bifur-Hopf}
\end{figure}

\vspace{1.8mm}

\begin{theorem}[Guckenheimer and Holmes]
Suppose that the heterogeneous actSIS model \eqref{eq:two-node} expressed as  $\dot{\boldsymbol{p}}=\boldsymbol{f}(\boldsymbol{p}, \Bar{\beta})$, $\boldsymbol{p} = (p_1, p_2, p_{s_1},p_{s_2})$, 
$\Bar{\beta} \in \mathbb{R}$, has an equilibrium at
$\left(\boldsymbol{p}^{\mathbf{*}}, \Bar{\beta}^*\right)$ and the following properties are satisfied:\\
\begin{itemize}
    \item ($\mathrm{H} 1$) \; The Jacobian $\left.D_{p} \boldsymbol{f}\right|_{\left(\boldsymbol{p}^{\mathbf{*}}, \Bar{\beta}^*\right)}$ has a simple pair of pure imaginary eigenvalues $\lambda\left(\Bar{\beta}^*\right)$ and $\overline{\lambda\left(\Bar{\beta}^*\right)}$ and no other eigenvalues with zero real parts,
    \item $\left.(\mathrm{H} 2) \; \frac{d}{d \Bar{\beta}}(\operatorname{Re} \lambda(\Bar{\beta}))\right|_{\left(\Bar{\beta}=\Bar{\beta}^*\right)} \neq 0$.
\end{itemize}
Then the dynamics undergo a Hopf bifurcation at $\left(\boldsymbol{p}^{\mathbf{*}}, \Bar{\beta}^*\right)$ resulting in periodic solutions. The stability of the periodic solutions is given by the sign of the first Lyapunov coefficient of the dynamics $\left.\ell_{1}\right|_{\left(\boldsymbol{p}^{\mathbf{*}}, \Bar{\beta}^*\right)} .$ If $\ell_{1}<0,$ then these solutions are stable limit cycles and the Hopf bifurcation is supercritical, while if $\ell_{1}>0,$ the periodic solutions are repelling.
\label{hopf-thm}
\end{theorem}

We show in Proposition~\ref{lemma:H1}  conditions on the model parameters that guarantee that the non-hyperbolicity condition (H1) is satisfied for the heterogeneous actSIS model \eqref{eq:two-node}. We verify  numerically that condition (H2) is satisfied and $\ell_{1}<0$.

\vspace{2mm} 

\begin{proposition} 
Denote
 $m=(1-p_1^*)(1-p_2^*)$,
$q=p_1^*\beta_{21}^*/p_2^*$,
 $s=p_2^*\beta_{12}^*/p_1^*$,
$v=p_1^*\beta_{21}^{\prime*}/p_2^*$,
$w=p_2^*\beta_{12}^{\prime*}/p_1^*$,
$\beta_{12}^*=\Bar{\beta}\alpha^{1}\left(p_{2}^*\right),$  
$\beta_{21}^*=\Bar{\beta}\alpha^{2}\left(p_{1}^*\right))$. Then, for system \eqref{eq:two-node}, the non-hyperbolicity condition (H1) in
 Theorem~\ref{hopf-thm} is satisfied if 
 \begin{equation*}
     c>0, \quad a \neq 0,
 \end{equation*}
 where 
 \begin{equation*}
     \begin{aligned}
      a&:=s+q+\frac{2}{\tau_s}, \\ ac&:=\frac{2}{\tau_s}(1-m)sq+\frac{1}{\tau_s^2}(s+q)
      -m(v\beta_{12}^* + w\beta_{21}^*).\\
     \end{aligned}
     \label{hopf:condition}
 \end{equation*}
\label{lemma:H1}
\end{proposition}
\begin{proof}
For a four-dimensional system to satisfy (H1) \\ the eigenvalues of the Jacobian $\left.D_{p} \boldsymbol{f}\right|_{\left(\boldsymbol{p}^{\mathbf{*}}, \Bar{\beta}^*\right)}$ must satisfy
\begin{equation}
    \left(\lambda_{}^2+c\right)\left(\lambda_{}^2+a\lambda_{}+b\right)=0,
\label{eq:H1-cond}
\end{equation}
for some $a\neq0$,$b\in \mathbb{R}$,and $c>0$. 
We compute  the Jacobian\\ 
\vspace{2mm} 
\resizebox{.5\textwidth}{!}{%
$\displaystyle
    \left.D_{p} \boldsymbol{f}\right|_{\left(\boldsymbol{p}^{\mathbf{*}}, \Bar{\beta}^*\right)}= \begin{bmatrix}
             - \beta_{12}^*p_2^*-\delta & (1-p_1^*)\beta_{12}^*  &0 & (1-p_1^*)p_2^*\beta_{12}^{\prime*} \\
             (1-p_2^*)\beta_{21}^* & - \beta_{21}^*p_1^*-\delta & (1-p_2^*)p_1^*\beta_{21}^{\prime*} &0\\
             \frac{1}{\tau_s}  & 0 & -\frac{1}{\tau_s} & 0\\
            0 & \frac{1}{\tau_s}  & 0 & -\frac{1}{\tau_s} \\
            \end{bmatrix}
$} which has eigenvalues that satisfy

\begin{equation}
\begin{aligned}
 \left(\frac{1}{\tau_s}+\lambda_{}\right)^2 \left(s+\lambda_{}\right)\left(q+\lambda_{}\right)
    -m\left(\frac{1}{\tau_s}+\lambda_{}\right)^2sq & \\
  -m\left(\frac{1}{\tau_s}+\lambda_{}\right)v\beta_{12}^*
    -m\left(\frac{1}{\tau_s}+\lambda_{}\right)w\beta_{21}^*
    -mvw   & =0.  
    \end{aligned}
        \label{eq:compute-eigen}
\end{equation}

Matching coefficients of \eqref{eq:H1-cond} and \eqref{eq:compute-eigen}, we derive the expressions for $a$, $b$, and $c$, in terms of the model parameters  ($\bar{\beta},\mu_{T}, \gamma_{T}, \mu_{A}, \gamma_{A}$) for  \eqref{eq:two-node} to satisfy the first part of (H1). 
We have $\lambda\left(\Bar{\beta}^*\right)=\sqrt{c} \di$, a pure imaginary eigenvalue. $a\neq0$ guarantees  there are no other eigenvalues with zero real part.
\end{proof}


A proof of conditions guaranteeing (H2) of Theorem~\ref{hopf-thm} is the subject of ongoing work. Fig \ref{fig:bifur-Hopf}(b) shows numerically that (H2) is satisfied for the parameters selected. We also  checked that  $\ell_1<0$. 
An illustration of the sustained oscillations corresponding to the stable limit cycle of \eqref{eq:two-node} is depicted in Fig~\ref{fig:osci}.


\begin{figure}[!htbp]
\centering
\subfigure[Sustained oscillations]{
    \includegraphics[width=0.45\textwidth]{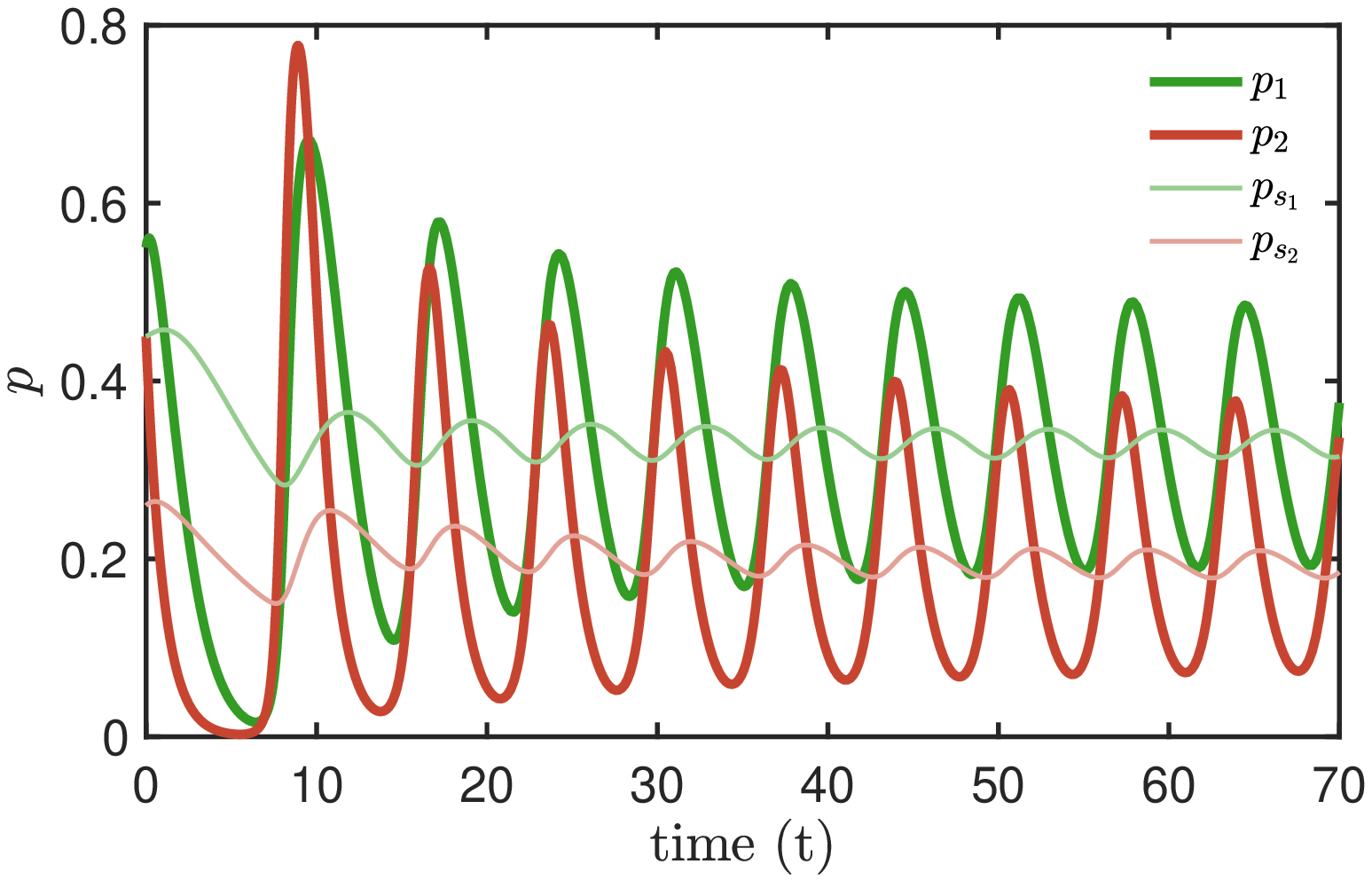}
}


\subfigure[$\beta_{12}$ evolution]{
    \includegraphics[width=0.22\textwidth]{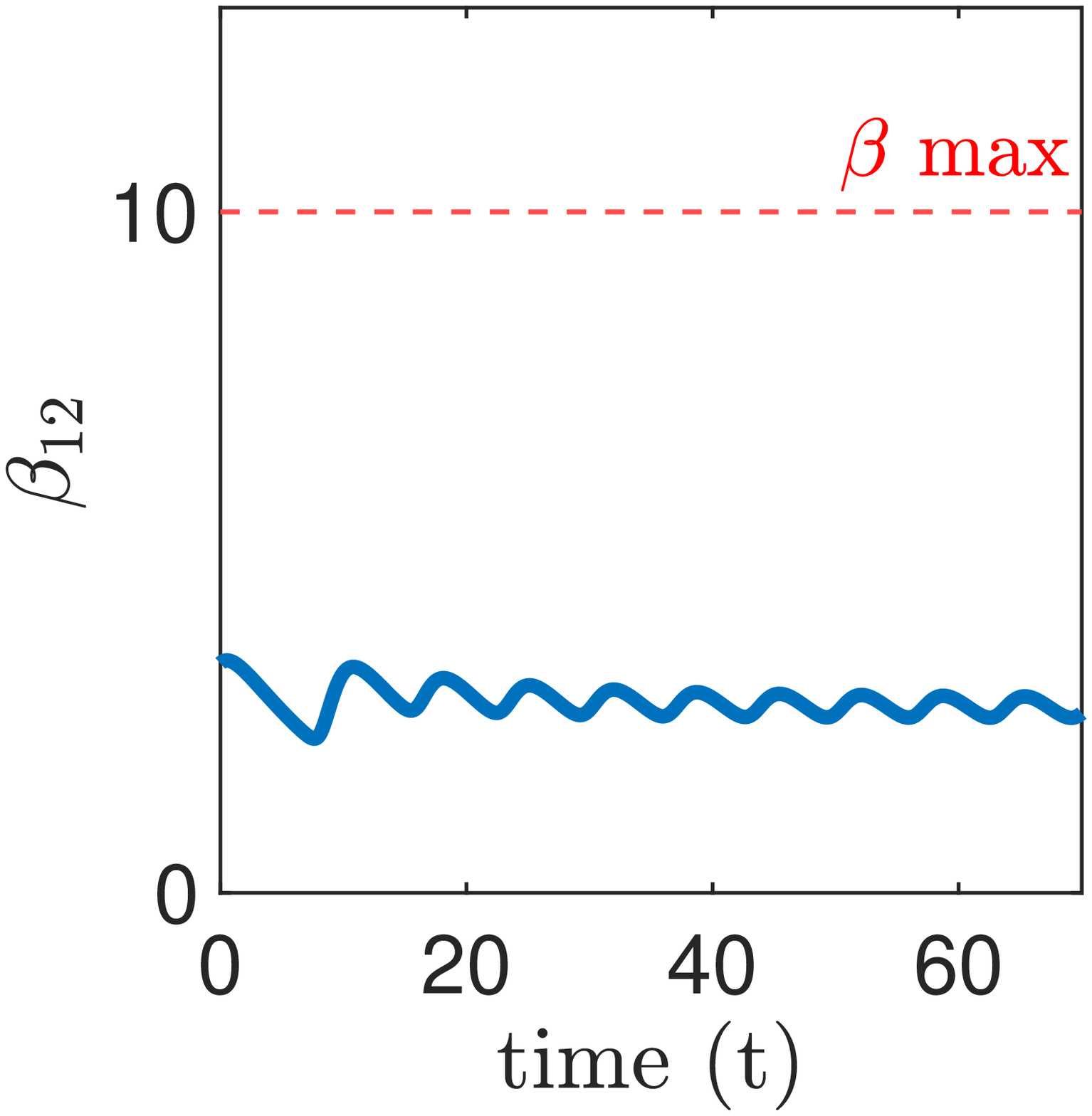}
}
\subfigure[$\beta_{21}$ evolution]{
    \includegraphics[width=0.22\textwidth]{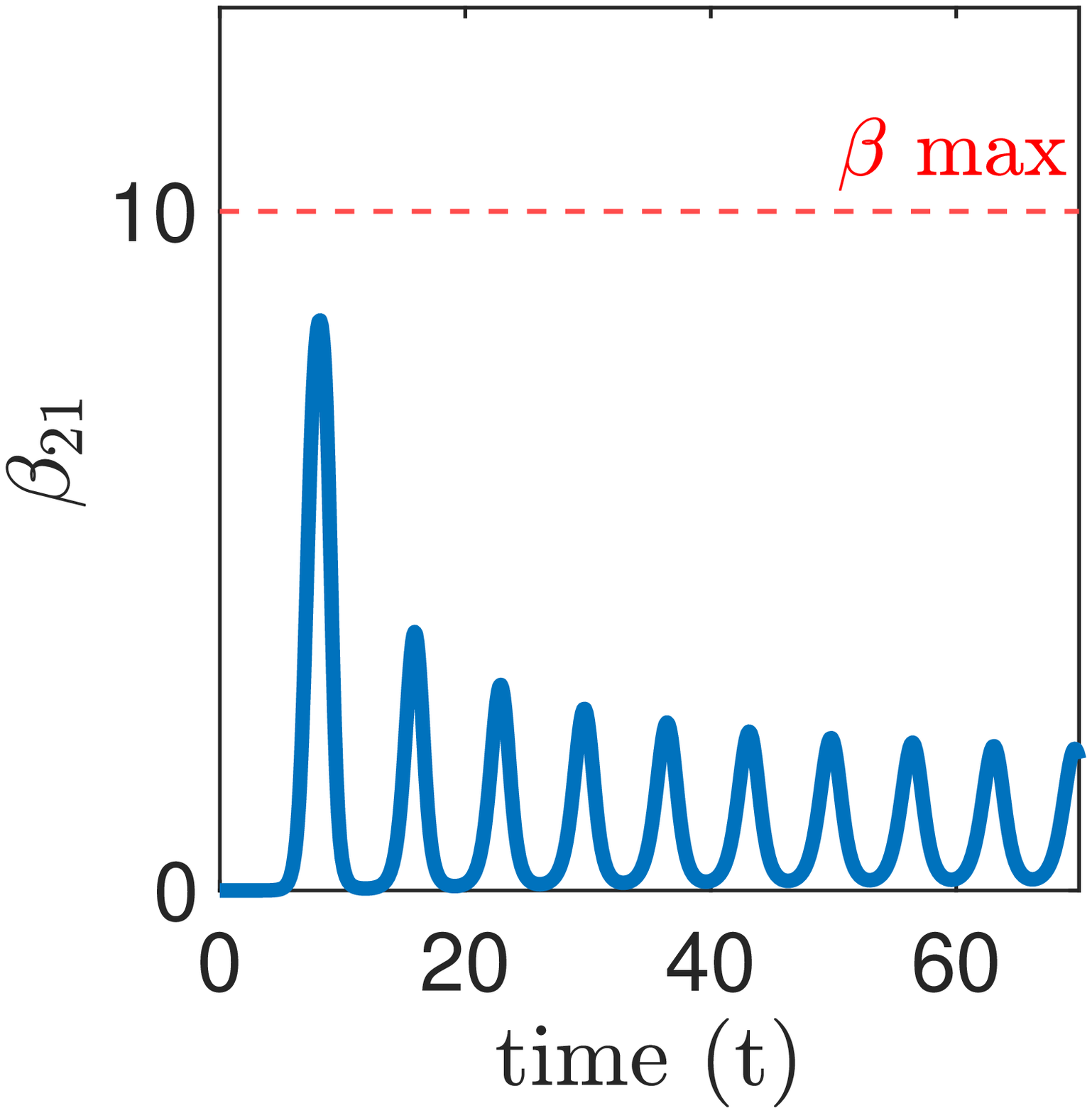}
}
  \centering
\figurecaption{%
Stable limit cycles of the heterogeneous actSIS model}
{%

a) When $\bar{\beta} >\bar{\beta}^* $, system \eqref{eq:two-node} exhibits stable oscillations. 
States $p_1$ and $p_2$  initially decrease quickly  
due to different mechanisms: 
$p_1$ decreases because $p_{s_2}(0) < \mu_T$, which drives $\beta_{12}$ down and  $p_2$ decreases because $p_{s_1}(0) > \mu_A$, which drives $\beta_{21}$ down. 
The low $p_1$ leads to a  decrease in $p_{s_1}$, which in turn drives up  $\beta_{21}$, thus increasing $p_2$. The increase in $p_2$ leads to an increase in $p_{s_2}$, which in turn drives up $p_1$. 
The process repeats resulting in  sustained oscillations. 
b-c): The time evolution of the effective infection rates. Parameters used for the simulation:  $\bar{\beta}= 8.8$, $\mu_T=0.45, \nu_T=0.8, \mu_A=0.3, \nu_A=20$. Initial conditions:  $p_1(0)=0.55, p_2(0)=0.45, p_{s_1}(0)=0.45, p_{s_2}(0)=0.66.$
}
\label{fig:osci}
\end{figure}

\section{Final remarks}
\label{sec:disc}

The actSIS model incorporates two novel mechanisms as compared with the SIS model: a feedback mechanism for the effective infection rates, and a time scale separation between the state of the system and the state used in the feedback law. The qualitative differences we have shown for the actSIS model are due to both of the mechanisms. We have observed sustained oscillations even when some of the individuals are \textit{risk-ignorers} and under relaxed assumptions on interconnections.  We will examine the broader set of possibilities in future work and consider applications in other biological and socio-ecological processes.





\end{document}